\documentclass{article}
\usepackage{amsfonts}
\usepackage{amssymb}
\usepackage[dvips]{color}
\usepackage{graphicx}
\usepackage{amsmath}

\newtheorem{theorem}{Theorem}

\newtheorem{corollary}[theorem]{Corollary}

\newtheorem{definition}[theorem]{Definition}
\newtheorem{example}[theorem]{Example}

\newtheorem{lemma}[theorem]{Lemma}

\newtheorem{proposition}[theorem]{Proposition}
\newtheorem{remark}[theorem]{Remark}

\newenvironment{proof}[1][Proof]{\textbf{#1.} }{\ \rule{0.5em}{0.5em}}

\input{tcilatex}

\begin{document}

\title{Risk Measures on $\mathcal{P}(\mathbb{R})$ and Value At Risk with
Probability/Loss function}
\author{Marco Frittelli \\
{\small Milano University, email: marco.frittelli@unimi.it} \and Marco
Maggis \thanks{%
The author acknowledges the financial support provided by the European
Social Fund Grant.} \qquad \\
{\small Milano University, email: marco.maggis@unimi.it} \and Ilaria Peri \\
{\small Milano-Bicocca University, email: ilaria.peri@unimib.it}\qquad }
\maketitle

\begin{abstract}
We propose a generalization of the classical notion of the $V@R_{\lambda }$
that takes into account not only the probability of the losses, but the
balance between such probability and the amount of the loss. This is
obtained by defining a new class of law invariant risk measures based on an
appropriate family of acceptance sets. The $V@R_{\lambda }$ and other known
law invariant risk measures turn out to be special cases of our proposal. We
further prove the dual representation of Risk Measures on $\mathcal{P}(%
\mathbb{R}).$
\end{abstract}



\noindent \textbf{Acknowledgment: }We would like to thank Prof. Fabio
Bellini, University Milano Bicocca, Samuel Drapeau, Humboldt University, as
well as an anonymous referee for helpful discussion on this subject.

\section{Introduction}


We introduce a new class of law invariant risk measures $\Phi :\mathcal{P}(%
\mathbb{R})\rightarrow \mathbb{R}\cup \{+\infty \}$ that are directly
defined on the set $\mathcal{P}(\mathbb{R})$ of probability measures on $%
\mathbb{R}$ and are monotone and quasi-convex on $\mathcal{P}(\mathbb{R})$. 

As Cherny and Madan (2009) \cite{CM09} pointed out, for a (\textit{%
translation invariant)} coherent risk measure defined on random variables,
all the positions can be spited in two classes: acceptable and not
acceptable; in contrast, for an \textit{acceptability index} there is a
whole continuum of degrees of acceptability defined by a system $\left\{ 
\mathcal{A}^{m}\right\} _{m\in \mathbb{R}}$ of sets. This formulation has
been further investigated by Drapeau and Kupper (2010) \cite{DK10} for the
quasi convex case, with emphasis on the notion of an acceptability family
and on the robust representation.

We adopt this approach and we build the maps $\Phi $ from a family $\left\{ 
\mathcal{A}^{m}\right\} _{m\in\mathbb{R}}$ of acceptance sets of
distribution functions by defining:%
\begin{equation*}
\Phi (P):=-\sup \left\{ m\in \mathbb{R}\mid P\in \mathcal{A}^{m}\right\} .
\end{equation*}

In Section 3 we study the properties of such maps, we provide some specific
examples and in particular we propose an interesting generalization of the
classical notion of $V@R_{\lambda }$.

\bigskip

The key idea of our proposal - the definition of the $\Lambda V@R$ in
Section 4 - arises from the consideration that in order to assess the risk
of a financial position it is necessary to consider not only the probability 
$\lambda $ of the loss, as in the case of the $V@R_{\lambda }$, but the
dependence between such \emph{probability} $\lambda $ and the \emph{amount}
of the loss. In other terms, a risk prudent agent is willing to accept
greater losses only with smaller probabilities. Hence, we replace the
constant $\lambda $ with a (increasing) function $\Lambda :\mathbb{%
R\rightarrow }[0,1]$ defined on losses, which we call \emph{Probability/Loss
function}. The balance between the probability and the amount of the losses
is incorporated in the definition of the family of acceptance sets 
\begin{equation*}
\mathcal{A}^{m}:=\left\{ Q\in \mathcal{P}(\mathbb{R})\mid Q(-\infty ,x]\leq
\Lambda (x),\;\forall x\leq m\right\} \text{, }m\in \mathbb{R}.
\end{equation*}%
If $P_{X}$ is the distribution function of the random variable $X,$ our new
measure is defined by:%
\begin{equation*}
\Lambda V@R(P_{X}):=-\sup \left\{ m\in \mathbb{R}\mid P(X\leq x)\leq \Lambda
(x),\;\forall x\leq m\right\} .
\end{equation*}%
As a consequence, the acceptance sets $\mathcal{A}^{m}$ are not obtained by
the translation of $\mathcal{A}^{0}$ which implies that the map is not any
more translation invariant. However, the similar property%
\begin{equation*}
\Lambda V@R(P_{X+\alpha })=\Lambda ^{\alpha }V@R(P_{X})-\alpha ,
\end{equation*}%
where $\Lambda ^{\alpha }(x)=\Lambda (x+\alpha )$, holds true and is
discussed in Section 4.

The $V@R_{\lambda }$ and the worst case risk measure are special cases of
the $\Lambda V@R$.

\bigskip

The approach of considering risk measures defined directly on the set of
distribution functions is not new and it was already adopted by Weber (2006) 
\cite{Weber}. However, in this paper we are interested in quasi-convex risk
measures based - as the above mentioned map $\Lambda V@R$ - on families of
acceptance sets of distributions and in the analysis of their robust
representation. We choose to define the risk measures on the entire set $%
\mathcal{P}(\mathbb{R})$ and not only on its subset of probabilities having
compact support, as it was done by Drapeau and Kupper (2010) \cite{DK10}.
For this, we endow $\mathcal{P}(\mathbb{R})$ with the $\sigma (\mathcal{P}(%
\mathbb{R}),C_{b}(\mathbb{R}))$ topology. The selection of this topology is
also justified by the fact (see Proposition \ref{CFB}) that for monotone
maps $\sigma (\mathcal{P}(\mathbb{R}),C_{b}(\mathbb{R}))-lsc$ is equivalent
to continuity from above. In section 5 we briefly compare the robust
representation obtained in this paper and those obtained by Cerreia-Vioglio
(2009) \cite{CV} and Drapeau and Kupper (2010) \cite{DK10}.

Except for $\Phi =+\infty $, we show that there are no \textit{convex}, $%
\sigma (\mathcal{P}(\mathbb{R}),C_{b}(\mathbb{R}))-lsc$ translation
invariant maps $\Phi :\mathcal{P}(\mathbb{R})\rightarrow \mathbb{R}\cup
\{+\infty \}$. But there are many quasi-convex and $\sigma (\mathcal{P}(%
\mathbb{R}),C_{b}(\mathbb{R}))-lsc$ maps $\Phi :\mathcal{P}(\mathbb{R}%
)\rightarrow \mathbb{R}\cup \{+\infty \}$ that in addition are monotone and
translation invariant, as for example the $V@R_{\lambda }$, the entropic
risk measure and the worst case risk measure. This is another good
motivation to adopt quasi convexity versus convexity.

Finally we provide the dual representation of quasi-convex, monotone and $%
\sigma (\mathcal{P}(\mathbb{R}),C_{b}(\mathbb{R}))-lsc$ maps $\Phi :\mathcal{%
P}(\mathbb{R})\rightarrow \mathbb{R}\cup \{+\infty \}$ - \emph{defined on
the entire set} $\mathcal{P}(\mathbb{R})$ - and compute the dual
representation of the risk measures associated to families of acceptance
sets and consequently of the $\Lambda V@R$.

\section{Law invariant Risk Measures}


Let $(\Omega ,\mathcal{F},\mathbb{P})$ be a probability space and $%
L^{0}=:L^{0}(\Omega ,\mathcal{F},\mathbb{P})$ be the space of $\mathcal{F}$
measurable random variables that are $\mathbb{P}$ almost surely finite. 
\newline
Any random variable $X\in L^{0}$ induces a probability measure $P_{X}$ on $(%
\mathbb{R}$,$\mathcal{B}_{\mathbb{R}})$ by $P_{X}(B)=\mathbb{P}(X^{-1}(B))$
for every Borel set $B\in \mathcal{B}_{\mathbb{R}}$. We refer to \cite{Ali}
Chapter 15 for a detailed study of the convex set $\mathcal{P}=:\mathcal{P}(%
\mathbb{R})$ of probability measures on $\mathbb{R}$. Here we just recall
some basic notions: for any $X\in L^{0}$ we have $P_{X}\in \mathcal{P}$ so
that we will associate to any random variable a unique element in $\mathcal{P%
}$. If $\mathbb{P}(X=x)=1$ for some $x\in \mathbb{R}$ then $P_{X}$ is the
Dirac distribution $\delta _{x}$ that concentrates the mass in the point $x$%
. \newline
A map $\rho :L\rightarrow \overline{\mathbb{R}}:=\mathbb{R}\cup \left\{
-\infty \right\} \cup \left\{ \infty \right\} $, defined on given subset $%
L\subseteq L^{0},$ is law invariant if $X,Y\in L$ and $P_{X}=P_{Y}$ implies $%
\rho (X)=\rho (Y)$.

Therefore, when considering law invariant risk measures $\rho
:L^{0}\rightarrow \overline{\mathbb{R}}$ it is natural to shift the problem
to the set $\mathcal{P}$ by defining the new map $\Phi :\mathcal{P}%
\rightarrow \overline{\mathbb{R}}$ \ as $\Phi (P_{X})=\rho (X)$. This map $%
\Phi $ is well defined on the entire $\mathcal{P}$, since there exists a
bi-injective relation between $\mathcal{P}$ and the quotient space $\frac{%
L^{0}}{\sim }$ (provided that $(\Omega ,\mathcal{F},\mathbb{P})$ supports a
random variable with uniform distribution), where the equivalence is given
by $X\sim _{\mathcal{D}}Y$ $\Leftrightarrow P_{X}=P_{Y}$. However, $\mathcal{%
P}$ is only a convex set and the usual operations on $\mathcal{P}$ are not
induced by those on $L^{0}$, namely $(P_{X}+P_{Y})(A)=P_{X}(A)+P_{Y}(A)\neq
P_{X+Y}(A)$, $A\in \mathcal{B}_{\mathbb{R}}$. 

Recall that the first order stochastic dominance on $\mathcal{P}$ is given
by: $Q\preccurlyeq P\Leftrightarrow F_{P}(x)\leq F_{Q}(x)$ for all $x\in 
\mathbb{R},$ where $F_{P}(x)=P(-\infty ,x]$ and $F_{Q}(x)=Q(-\infty ,x]$ are
the distribution functions of $P,Q\in \mathcal{P}$. Notice that $X\leq Y$ $%
\mathbb{P}$-a.s. implies $P_{X}\preccurlyeq P_{Y}$.

\begin{definition}
A Risk Measure on $\mathcal{P}(\mathbb{R} )$ is a map $\Phi :\mathcal{P}%
\rightarrow \mathbb{R}\cup \{+\infty \}$ such that:

\begin{description}
\item[(Mon)] $\Phi $ is monotone decreasing: $P\preccurlyeq Q$ implies $\Phi
(P)\geq \Phi (Q)$;

\item[(QCo)] $\Phi $ is quasi-convex: $\Phi (\lambda P+(1-\lambda )Q)\leq
\Phi (P)\vee \Phi (Q)$, $\lambda \in \lbrack 0,1].$
\end{description}
\end{definition}

Quasiconvexity can be equivalently reformulated in terms of sublevel sets: a
map $\Phi $ is quasi-convex if for every $c\in \mathbb{R}$ the set $\mathcal{%
A}_{c}=\{P\in \mathcal{P}\mid \Phi (P)\leq c\}$ is convex. As recalled in 
\cite{Weber} this notion of convexity is different from the one given for
random variables (as in \cite{FoSch}) because it does not concern
diversification of financial positions. A natural interpretation in terms of
compound lotteries is the following: whenever two probability measures $P$
and $Q$ are acceptable at some level $c$ and $\lambda \in \lbrack 0,1]$ is a
probability, then the compound lottery $\lambda P+(1-\lambda )Q$, which
randomizes over $P$ and $Q$, is also acceptable at the same level. \newline
In terms of random variables (namely $X,Y$ which induce $P_{X},P_{Y}$) the
randomized probability $\lambda P_{X}+(1-\lambda )P_{Y}$ will correspond to
some random variable $Z\neq \lambda X+(1-\lambda )Y$ so that the
diversification is realized at the level of distribution and not at the
level of portfolio selection.

\bigskip

As suggested by \cite{Weber}, we define the translation operator $T_{m}$ on
the set $\mathcal{P}(\mathbb{R})$ by: $T_{m}P(-\infty ,x]=P(-\infty ,x-m]$,
for every $m\in \mathbb{R}$. Equivalently, if $P_{X}$ is the probability
distribution of a random variable $X$ we define the translation operator as $%
T_{m}P_{X}=P_{X+m}$, $m\in \mathbb{R}$. As a consequence we map the
distribution $F_{X}(x)$ into $F_{X}(x-m)$. Notice that $P\preccurlyeq T_{m}P$
for any $m>0$.

\begin{definition}
If $\Phi :\mathcal{P}\rightarrow \mathbb{R}\cup \{+\infty \}$ is a risk
measure on $\mathcal{P}$, we say that

\begin{description}
\item[(TrI)] $\Phi $ is translation invariant if $\Phi (T_{m}P)=\Phi (P)-m$
for any $m\in \mathbb{R}.$
\end{description}
\end{definition}

Notice that (TrI) corresponds exactly to the notion of cash additivity for
risk measures defined on a space of random variables as introduced in \cite%
{ADEH}. It is well known (see \cite{CMMMa}) that for maps defined on random
variables, quasiconvexity and cash additivity imply convexity. However, in
the context of distributions (QCo) and (TrI) do not imply convexity of the
map $\Phi $, as can be shown with the simple examples of the $V@R$ and the
worst case risk measure $\rho _{w}$ (see the examples in Section 3.1).

The set $\mathcal{P}(\mathbb{R})$ spans the space $ca(\mathbb{R}):=\{\mu 
\text{ signed measure }\mid V_{\mu }<+\infty \}$ of all signed measures of
bounded variations on $\mathbb{R}$. $ca(\mathbb{R})$ (or simply $ca$)
endowed with the norm $V_{\mu }=\sup \left\{ \sum_{i=1}^{n}|\mu (A_{i})| 
\text{ s.t. } \{A_1,...,A_{n}\}\text{ partition of }\mathbb{R}\right\} $ is
a norm complete and an AL-space (see \cite{Ali} paragraph 10.11).

Let $C_{b}(\mathbb{R})$ (or simply $C_{b}$) be the space of bounded
continuous function $f:\mathbb{R}\rightarrow \mathbb{R}$. We endow $ca(%
\mathbb{R})$ with the weak$^{\ast }$ topology $\sigma (ca,C_{b})$. The dual
pairing $\langle \cdot ,\cdot \rangle :C_{b}\times ca\rightarrow \mathbb{R}$
is given by $\langle f,\mu \rangle =\int fd\mu $ and the function $\mu
\mapsto \int fd\mu $ ($\mu \in ca$) is $\sigma (ca,C_{b})$ continuous.
Notice that $\mathcal{P}$ is a $\sigma (ca,C_{b})$-closed convex subset of $%
ca$ (p. 507 in \cite{Ali}) so that $\sigma (\mathcal{P},C_{b})$ is the
relativization of $\sigma (ca,C_{b})$ to $\mathcal{P}$ and any $\sigma (%
\mathcal{P},C_{b})$-closed subset of $\mathcal{P}$ is also $\sigma
(ca,C_{b}) $-closed.

Even though $(ca,\sigma (ca,C_{b}))$ is not metrizable in general, its
subset $\mathcal{P}$ is separable and metrizable (see \cite{Ali}, Th.15.12)
and therefore when dealing with convergence in $\mathcal{P}$ we may work
with sequences instead of nets.

For every real function $F$ we denote by $\mathcal{C}(F)$ the set of points
in which the function $F$ is continuous.

\begin{theorem}
\label{weak}(\cite{Shy} Theorem 2, p.314) ) Suppose that $P_{n}$, $P\in 
\mathcal{P}$. Then $P_{n}\overset{\sigma (\mathcal{P},C_{b})}{%
\longrightarrow }P$ if and only if $F_{P_{n}}(x)\rightarrow F_{P}(x)$ for
every $x\in \mathcal{C}(F_{P})$.
\end{theorem}


A sequence of probabilities $\left\{ P_{n}\right\} \subset \mathcal{P}$ is
decreasing, denoted with $P_{n}\downarrow $, if $F_{P_{n}}(x)\leq
F_{P_{n+1}}(x)$ for all $x\in \mathbb{R}$ and all $n$.

\begin{definition}
Suppose that $P_{n}$, $P\in \mathcal{P}$. We say that $P_{n}\downarrow P$
whenever $P_{n}\downarrow $ and $F_{P_{n}}(x)\uparrow F_{P}(x)$ for every $%
x\in \mathcal{C}(F_{P})$. We say that

\begin{description}
\item[(CfA)] $\Phi $ is continuous from above if $P_{n}\downarrow P$ implies 
$\Phi (P_{n})\uparrow \Phi (P).$
\end{description}
\end{definition}

\begin{proposition}
\label{CFB}Let $\Phi :\mathcal{P}\rightarrow \overline{\mathbb{R}}$ be
(Mon). Then the following are equivalent:

$\Phi $ is $\sigma (\mathcal{P},C_{b})$-lower semicontinuous

$\Phi $ is continuous from above.
\end{proposition}

\begin{proof}
Let $\Phi $ be $\sigma (\mathcal{P},C_{b})$-lower semicontinuous and suppose
that $P_{n}\downarrow P$. Then $F_{P_{n}}(x)\uparrow F_{P}(x)$ for every $%
x\in \mathcal{C}(F_{P})$ and we deduce from Theorem \ref{weak} that $P_{n}%
\overset{\sigma (\mathcal{P},C_{b})}{\longrightarrow }P$. (Mon) implies $%
\Phi (P_{n})\uparrow $ and $k:=\lim_{n}\Phi (P_{n})\leq \Phi (P)$. The lower
level set $A_{k}=\{Q\in \mathcal{P}\mid \Phi (Q)\leq k\}$ is $\sigma (%
\mathcal{P},C_{b})$ closed and, since $P_{n}\in A_{k},$ we also have $P\in
A_{k},$ i.e. $\Phi (P)=k$, and $\Phi $ is continuous from above.

Conversely, suppose that $\Phi $ is continuous from above. As $\mathcal{P}$
is metrizable we may work with sequences instead of nets. For $k\in \mathbb{R%
}$ consider $A_{k}=\{P\in \mathcal{P}\mid \Phi (P)\leq k\}$ and a sequence $%
\{P_{n}\}\subseteq A_{k}$ such that $P_{n}\overset{\sigma (\mathcal{P},C_{b})%
}{\longrightarrow }P\in \mathcal{P}$. We need to show that $P\in A_{k}.$
Lemma \ref{LCFA} shows that each $F_{Q_{n}}:=(\inf_{m\geq n}F_{P_{m}})\wedge
F_{P}$ is the distribution function of a probability measure and $%
Q_{n}\downarrow P$. From (Mon) and $P_{n}\preccurlyeq Q_{n}$, we get $\Phi
(Q_{n})\leq \Phi (P_{n}).$ From (CfA) then: $\Phi (P)=\lim_{n}\Phi
(Q_{n})\leq \liminf_{n}\Phi (P_{n})\leq k$. Thus $P\in A_{k}.$
\end{proof}

\begin{lemma}
\label{LCFA}For every $P_{n}\overset{\sigma (\mathcal{P},C_{p})}{%
\longrightarrow }P$ we have that 
\begin{equation*}
F_{Q_{n}}:=\inf_{m\geq n}F_{P_{m}}\wedge F_{P}\text{, }n\in \mathbb{N}\text{,%
}
\end{equation*}%
is a distribution function associated to a probability measure $Q_{n}$ $\in 
\mathcal{P}$ such that $Q_n\downarrow P$.
\end{lemma}

\begin{proof}
For each $n,$ $F_{Q_{n}}$ is increasing and $\lim_{x\rightarrow -\infty
}F_{Q_{n}}(x)=0.$ Moreover for real valued maps right continuity and upper
semicontinuity are equivalent. Since the $\inf $-operator preserves upper
semicontinuity we can conclude that $F_{Q_{n}}$ is right continuous for
every $n$. Now we have to show that for each $n$, $\lim_{x\rightarrow
+\infty }F_{Q_{n}}(x)=1$. By contradiction suppose that, for some $n$, $%
\lim_{x\rightarrow +\infty }F_{Q_{n}}(x)=\lambda <1$. We can choose a
sequence $\{x_{k}\}_{k}\subseteq \mathbb{R}$ with $x_{k}\in \mathcal{C}%
(F_{P})$, $x_{k}\uparrow +\infty $. In particular $F_{Q_{n}}(x_{k})\leq
\lambda $ for all $k$ and $F_{P}(x_{k})>\lambda $ definitively, say for all $%
k\geq k_{0}$. We can observe that since $x_{k}\in \mathcal{C}(F_{P})$ we
have, for all $k\geq k_{0}$, $\inf_{m\geq
n}F_{P_{m}}(x_{k})<\lim_{m\rightarrow +\infty }F_{P_{m}}(x_{k})=F_{P}(x_{k})$%
. This means that the infimum is attained for some index $m(k)\in \mathbb{N}$%
, i.e. $\inf_{m\geq n}F_{P_{m}}(x_{k})=F_{P_{m(k)}}(x_{k})$, for all $k\geq
k_{0}$. Since $P_{m(k)}(-\infty ,x_{k}]=F_{P_{m(k)}}(x_{k})\leq \lambda $
then $P_{m(k)}(x_{k},+\infty )\geq 1-\lambda $ for $k\geq k_{0}$. We have
two possibilities. Either the set $\{m(k)\}_{k}$ is bounded or $\overline{%
\lim }_{k}m(k)=+\infty $. In the first case, we know that the number of $m(k)
$'s is finite. Among these $m(k)$'s we can find at least one $\overline{m}$
and a subsequence $\left\{ x_{h}\right\} _{h}$ of $\{x_{k}\}_{k}$ such that $%
x_{h}\uparrow +\infty $ and $P_{\overline{m}}(x_{h},+\infty )\geq 1-\lambda $
for every $h$. We then conclude that 
\begin{equation*}
\lim_{h\rightarrow +\infty }P_{\overline{m}}(x_{h},+\infty )\geq 1-\lambda
\end{equation*}%
and this is a contradiction. If $\overline{\lim }_{k}m(k)=+\infty ,$ fix $%
\overline{k}\geq k_{0}$ such that $P(x_{\overline{k}},+\infty )<1-\lambda $
and observe that for every $k>\overline{k}$ 
\begin{equation*}
P_{m(k)}(x_{\overline{k}},+\infty )\geq P_{m(k)}(x_{k},+\infty )\geq
1-\lambda .
\end{equation*}%
Take a subsequence $\left\{ m(h)\right\} _{h}$ of $\{m(k)\}_{k}$ such that $%
m(h)\uparrow +\infty $. Then: 
\begin{equation*}
\lim_{h\rightarrow \infty }\inf P_{m(h)}(x_{\overline{k}},+\infty )\geq
1-\lambda >P(x_{\overline{k}},+\infty ),
\end{equation*}%
which contradicts the weak convergence $P_{n}\overset{\sigma (\mathcal{P}%
,C_{b})}{\longrightarrow }P$. \newline
Finally notice that $F_{Q_{n}}\leq F_{P_{n}}$ and $Q_{n}\downarrow $. From $%
P_{n}\overset{\sigma (\mathcal{P},C_{b})}{\longrightarrow }P$ and the
definition of $Q_{n}$, we deduce that $F_{Q_{n}}(x)\uparrow F_{P}(x)$ for
every $x\in \mathcal{C}(F_{P})$ so that $Q_{n}\downarrow P$.
\end{proof}

\begin{example}[The certainty equivalent]
It is very simple to build risk measures on $\mathcal{P}(\mathbb{R}).$ Take
any continuous, bounded from below and strictly decreasing function $f:%
\mathbb{R}\rightarrow \mathbb{R}$. Then the map $\Phi _{f}:\mathcal{P}%
\rightarrow \mathbb{R}\cup \{+\infty \}$ defined by: 
\begin{equation}
\Phi _{f}(P):=-f^{-1}\left( \int fdP\right)  \label{ceq}
\end{equation}%
is a Risk Measure on $\mathcal{P}(\mathbb{R})$. It is also easy to check
that $\Phi _{f}$ is (CfA) and therefore $\sigma (\mathcal{P},C_{b})-$lsc
Notice that Proposition \ref{NO} will then imply that $\Phi _{f}$ can not be
convex. By selecting the function $f(x)=e^{-x}$ we obtain $\Phi _{f}(P)=\ln
\left( \int \exp \left( -x)dF_{P}(x)\right) \right) $, which is in addition
(TrI). Its associated risk measure $\rho :L^{0}\rightarrow \mathbb{R}\cup
\{+\infty \}$ defined on random variables, $\rho (X)=\Phi _{f}(P_{X})=\ln
\left( Ee^{-X}\right) ,$\ is the Entropic (convex) Risk Measure. In Section
5 we will see more examples based on this construction.
\end{example}

\section{A remarkable class of risk measures on $\mathcal{P}(\mathbb{R} )$}


Given a family $\left\{ F_{m}\right\} _{m\in \mathbb{R}}$ of functions $%
F_{m}:\mathbb{R}\rightarrow \lbrack 0,1]$, we consider the associated sets
of probability measures 
\begin{equation}
\mathcal{A}^{m}:=\{Q\in \mathcal{P}\mid F_{Q}\leq F_{m}\}  \label{am}
\end{equation}%
and the associated map $\Phi :\mathcal{P}\rightarrow \overline{\mathbb{R}}$
defined by 
\begin{equation}
\Phi (P):=-\sup \left\{ m\in \mathbb{R}\mid P\in \mathcal{A}^{m}\right\} .
\label{phi1}
\end{equation}%
We assume hereafter that for each $P\in \mathcal{P}$ there exists $m$ such
that $P\notin \mathcal{A}^{m}$ so that $\Phi :\mathcal{P}\rightarrow \mathbb{%
R}\cup \{+\infty \}.$

\bigskip 

Notice that $\Phi (P):=\inf \left\{ m\in \mathbb{R}\mid P\in A_{m}\right\} $
where $A_{m}=:\mathcal{A}^{-m}$ and $\Phi (P)$ can be interpreted as the
minimal risk acceptance level under which $P$ is still acceptable. The
following discussion will show that under suitable assumption on $\left\{
F_{m}\right\} _{m\in \mathbb{R}}$ we have that $\left\{ A_{m}\right\} _{m\in 
\mathbb{R}}$ is a risk acceptance family as defined in \cite{DK10}.

We recall from \cite{DK10} the following definition

\begin{definition}
A monotone decreasing family of sets $\left\{ \mathcal{A}^{m}\right\} _{m\in 
\mathbb{R}}$ contained in $\mathcal{P}$ is \emph{left continuous} in $m$ if 
\begin{equation*}
\mathcal{A}^{m}=:\bigcap_{\varepsilon >0}\mathcal{A}^{m-\varepsilon }
\end{equation*}%
In particular it is \emph{left continuous} if it is \emph{left continuous}
in $m$ for every $m\in \mathbb{R}$.
\end{definition}

\begin{lemma}
\label{lemma2}Let $\left\{ F_{m}\right\} _{m\in \mathbb{R}}$ be a family of
functions $F_{m}:\mathbb{R}\rightarrow \lbrack 0,1]$ and $\mathcal{A}^{m}$
be the set defined in (\ref{am}). Then:

\begin{enumerate}
\item If, for every $x\in \mathbb{R}$, $F_{\centerdot }(x)$ is decreasing
(w.r.t. $m$) then the family $\left\{ \mathcal{A}^{m}\right\} $ is monotone
decreasing: $\mathcal{A}^{m}\subseteq \mathcal{A}^{n}$ for any level $m\geq
n,$

\item For any $m$, $\mathcal{A}^{m}$ is convex and satisfies: $Q\preceq P\in 
\mathcal{A}^{m}\Rightarrow Q\in \mathcal{A}^{m}$

\item If, for every $m\in \mathbb{R}$, $F_{m}(x)$ is right continuous w.r.t. 
$x$ then $\mathcal{A}^{m}$ is $\sigma (\mathcal{P},C_{b})-$closed,

\item Suppose that, for every $x\in \mathbb{R}$, $F_{m}(x)$ is decreasing
w.r.t. $m$. If $F_{m }(x)$ is left continuous w.r.t. $m$, then the family $\{%
\mathcal{A}^{m}\}$ is left continuous.

\item Suppose that, for every $x\in \mathbb{R}$, $F_{m }(x)$ is decreasing
w.r.t. $m$ and that, for every $m\in \mathbb{R}$, $F_{m}(x)$ is right
continuous and increasing w.r.t. $x$ and $\lim_{x\rightarrow +\infty
}F_{m}(x)=1$. If the family $\{\mathcal{A}^{m}\}$ is left continuous in $m$
then $F_{m}(x)$ is left continuous in $m$.
\end{enumerate}
\end{lemma}

\begin{proof}
1. If $Q\in \mathcal{A}^{m}$ and $m\geq n$ then $F_{Q}\leq F_{m}\leq F_{n}$,
i.e. $Q\in \mathcal{A}^{n}$.

\bigskip

2. Let $Q,P\in \mathcal{A}^{m}$ and $\lambda \in \left[ 0,1\right] $.
Consider the convex combination $\lambda Q+(1-\lambda )P$ and notice that 
\begin{equation*}
F_{\lambda Q+(1-\lambda )P}\leq F_{Q}\vee F_{P}\leq F_{m},
\end{equation*}%
as $F_{P}\leq F_{m}$ and $F_{Q}\leq F_{m}$. Then $\lambda Q+(1-\lambda )P\in 
\mathcal{A}^{m}$.

\bigskip

3. Let $Q_{n}\in A^{m}$ and $Q\in \mathcal{P}$ satisfy $Q_{n}$ $\overset{%
\sigma (\mathcal{P},C_{b})}{\rightarrow }Q$. By Theorem \ref{weak} we know
that $F_{Q_{n}}(x)\rightarrow F_{Q}(x)$ for every $x\in \mathcal{C}(F_{Q})$.
For each $n,$ $F_{Q_{n}}\leq F_{m}$ and therefore $F_{Q}(x)\leq F_{m}(x)$
for every $x\in \mathcal{C}(F_{Q})$. By contradiction, suppose that $Q\notin 
\mathcal{A}^{m}$. Then there exists $\bar{x}\notin \mathcal{C}(F_{Q})$ such
that $F_{Q}(\bar{x})>F_{m}(\bar{x})$. By right continuity of $F_{Q}$ for
every $\varepsilon >0$ we can find a right neighborhood $[\bar{x},\bar{x}%
+\delta (\varepsilon ))$ such that 
\begin{equation*}
|F_{Q}(x)-F_{Q}(\bar{x})|<\varepsilon \quad \forall \,x\in \lbrack \bar{x},%
\bar{x}+\delta (\varepsilon ))
\end{equation*}%
and we may require that $\delta (\varepsilon )\downarrow 0$ if $\varepsilon
\downarrow 0.$Notice that for each $\varepsilon >0$ we can always choose an $%
x_{\varepsilon }\in (\bar{x},\bar{x}+\delta (\varepsilon ))$ such that $%
x_{\varepsilon }\in \mathcal{C}(F_{Q})$. For such an $x_{\varepsilon }$ we
deduce that 
\begin{equation*}
F_{m}(\bar{x})<F_{Q}(\bar{x})<F_{Q}(x_{\varepsilon })+\varepsilon \leq
F_{m}(x_{\varepsilon })+\varepsilon .
\end{equation*}%
This leads to a contradiction since if $\varepsilon \downarrow 0$ we have
that $x_{\varepsilon }\downarrow \bar{x}$ and thus by right continuity of $%
F_{m}$ 
\begin{equation*}
F_{m}(\bar{x})<F_{Q}(\bar{x})\leq F_{m}(\bar{x}).
\end{equation*}%
4. By assumption we know that $F_{m-\varepsilon }(x)\downarrow F_{m}(x)$ as $%
\varepsilon \downarrow 0$, for all $x\in \mathbb{R}$. By item 1, we know
that $\mathcal{A}^{m}\subseteq \bigcap\limits_{\varepsilon >0}\mathcal{A}%
^{m-\varepsilon }$. By contradiction we suppose that the strict inclusion 
\begin{equation*}
\mathcal{A}^{m}\subset \bigcap\limits_{\varepsilon >0}\mathcal{A}%
^{m-\varepsilon }
\end{equation*}%
holds, so that there will exist $Q\in \mathcal{P}$ such that $F_{Q}\leq
F_{m-\varepsilon }$ for every $\varepsilon >0$ but $F_{Q}(\overline{x}%
)>F_{m}(\overline{x})$ for some $\overline{x}\in \mathbb{R}$. Set $\delta
=F_{Q}(\overline{x})-F_{m}(\overline{x})$ so that $F_{Q}(\overline{x})>F_{m}(%
\overline{x})+\frac{\delta }{2}$. Since $F_{m-\varepsilon }\downarrow F_{m}$
we may find $\overline{\varepsilon }>0$ such that $F_{m-\overline{%
\varepsilon }}(\overline{x})-F_{m}(\overline{x})<\frac{\delta }{2}$. Thus $%
F_{Q}(\overline{x})\leq F_{m-\varepsilon }(\overline{x})<F_{m}(\overline{x})+%
\frac{\delta }{2}$ and this is a contradiction.

5. Assume that $\mathcal{A}^{m-\varepsilon }\downarrow \mathcal{A}^{m}$.
Define $F(x):=\lim_{\varepsilon \downarrow 0}F_{m-\varepsilon
}(x)=\inf_{\varepsilon >0}F_{m-\varepsilon }(x)$ for all $x\in \mathbb{R}.$
Then $F:\mathbb{R}\rightarrow \lbrack 0,1]$ is increasing, right continuous
(since the $\inf $ preserves this property). Notice that for every $%
\varepsilon >0$ we have $F_{m-\varepsilon }\geq F\geq F_{m}$ and then $%
\mathcal{A}^{m-\varepsilon }\supseteq \{Q\in \mathcal{P}\mid F_{Q}\leq
F\}\supseteq \mathcal{A}^{m}$ and $\lim_{x\rightarrow +\infty }F(x)=1$.
Necessarily we conclude $\{Q\in \mathcal{P}\mid F_{Q}\leq F\}=\mathcal{A}%
^{m} $. By contradiction we suppose that $F(\overline{x})>F_{m}(\overline{x}%
) $ for some $\overline{x}\in \mathbb{R}$. Define $F_{\overline{Q}}:\mathbb{R%
}\rightarrow \lbrack 0,1]$ by: $F_{\overline{Q}}(x)=F(x)\mathbf{1}_{[%
\overline{x},+\infty )}(x)$. The above properties of $F$ guarantees that $F_{%
\overline{Q}}$ is a distribution function of a corresponding probability
measure $\overline{Q}\in \mathcal{P}$, and since $F_{\overline{Q}}\leq F$,
we deduce $\overline{Q}\in \mathcal{A}^{m}$, but $F_{\overline{Q}}(\overline{%
x})>F_{m}(\overline{x})$ and this is a contradiction.
\end{proof}

The following Lemma can be deduced directly from Lemma \ref{lemma2} and
Theorem 1.7 in \cite{DK10} (using the risk acceptance family $A_{m}=:%
\mathcal{A}^{-m}$, according to Definition 1.6 in the aforementioned paper).
We provide the proof for sake of completeness.

\begin{lemma}
\label{L3}Let $\left\{ F_{m}\right\} _{m\in \mathbb{R}}$ be a family of
functions $F_{m}:\mathbb{R}\rightarrow \lbrack 0,1]$ and $\Phi $ be the
associated map defined in (\ref{phi1}). Then:

\begin{enumerate}
\item The map $\Phi $ is (Mon) on $\mathcal{P}$.

\item If, for every $x\in \mathbb{R}$, $F_{\centerdot }(x)$ is decreasing
(w.r.t. $m$) then $\Phi $ is (QCo) on $\mathcal{P}$.

\item If, for every $x\in \mathbb{R}$, $F_{\centerdot }(x)$ is left
continuous and decreasing (w.r.t. $m$) and if, for every $m\in \mathbb{R}$, $%
F_{m}(\centerdot )$ is right continuous (w.r.t. $x$) then 
\begin{equation}
A_{m}:=\left\{ Q\in \mathcal{P}\mid \Phi (Q)\leq m\right\} =\mathcal{A}^{-m}%
\text{, }\forall m,  \label{am1}
\end{equation}%
and $\Phi $ is $\sigma (\mathcal{P},C_{b})-$lower-semicontinuous.
\end{enumerate}
\end{lemma}

\begin{proof}
1. From $P\preccurlyeq Q$ we have $F_{Q}\leq F_{P}$ and%
\begin{equation*}
\left\{ m\in 
\mathbb{R}
\mid F_{P}\leq F_{m}\right\} \subseteq \left\{ m\in 
\mathbb{R}
\mid F_{Q}\leq F_{m}\right\} ,
\end{equation*}%
which implies $\Phi (Q)\leq \Phi (P)$.

2. We show that $Q_{1},Q_{2}\in \mathcal{P}$, $\Phi (Q_{1})\leq n$ and $\Phi
(Q_{2})\leq n$ imply that $\Phi (\lambda Q_{1}+(1-\lambda )Q_{2})\leq n$,
that is 
\begin{equation*}
\sup \left\{ m\in 
\mathbb{R}
\mid F_{\lambda Q_{1}+(1-\lambda )Q_{2}}\leq F_{m}\right\} \geq -n.
\end{equation*}%
By definition of the supremum, $\forall \varepsilon >0$ $\exists m_{i}$ s.t. 
$F_{Q_{i}}\leq F_{m_{i}}$\ and $m_{i}>-\Phi (Q_{i})-\varepsilon \geq
-n-\varepsilon $. Then $F_{Q_{i}}\leq F_{m_{i}}\leq F_{-n-\varepsilon }$, as 
$\left\{ F_{m}\right\} $ is a decreasing family. Therefore $\lambda
F_{Q_{1}}+(1-\lambda )F_{Q_{2}}\leq F_{-n-\varepsilon }$ \ and $-\Phi
(\lambda Q_{1}+(1-\lambda )Q_{2}\lambda )\geq -n-\varepsilon $. As this
holds for any $\varepsilon >0$, we conclude that $\Phi $ is quasi-convex.

3. The fact that $\mathcal{A}^{-m}\subseteq A_{m}$ follows directly from the
definition of $\Phi ,$ as if $Q\in \mathcal{A}^{-m}$ 
\begin{equation*}
\Phi (Q):=-\sup \left\{ n\in \mathbb{R}\mid Q\in \mathcal{A}^{n}\right\}
=\inf \left\{ n\in \mathbb{R}\mid Q\in \mathcal{A}^{-n}\right\} \leq m.
\end{equation*}%
We have to show that $A_{m}\subseteq \mathcal{A}^{-m}$. Let $Q\in A_{m}$.
Since $\Phi (Q)\leq m$, for all $\varepsilon >0$ there exists $m_{0}$ such
that $m+\varepsilon >-m_{0}$ and $F_{Q}\leq F_{m_{0}}.$ Since $F_{\centerdot
}(x)$ is decreasing (w.r.t. $m$) we have that $F_{Q}\leq F_{-m-\varepsilon }$%
, therefore $Q\in \mathcal{A}^{-m-\varepsilon }$ for any $\varepsilon >0$.
By the left continuity in $m$ of $F_{\centerdot }(x),$ we know that$\{%
\mathcal{A}^{m}\}$ is left continuous (Lemma \ref{lemma2}, item 4) and so: $%
Q\in \bigcap\limits_{\epsilon >0}\mathcal{A}^{-m-\varepsilon }=\mathcal{A}%
^{-m}$.

From the assumption that $F_{m}(\centerdot )$ is right continuous (w.r.t. $x$%
) and Lemma \ref{lemma2} item 3, we already know that $\mathcal{A}^{m}$ is $%
\sigma (\mathcal{P},C_{b})-$closed, for any $m\in 
\mathbb{R}
$, and therefore the lower level sets $A_{m}=\mathcal{A}^{-m}$ are $\sigma (%
\mathcal{P},C_{b})-$closed and $\Phi $ is $\sigma (\mathcal{P},C_{b})-$%
lower-semicontinuous.
\end{proof}

\begin{definition}
A family $\left\{ F_{m}\right\} _{m\in \mathbb{R}}$ of functions $F_{m}:%
\mathbb{R}\rightarrow \lbrack 0,1]$ is \emph{feasible} if

\begin{itemize}
\item For any $P\in \mathcal{P}$ there exists $m$ such that $P\notin 
\mathcal{A}^{m}$

\item For every $m\in \mathbb{R}$, $F_{m}(\centerdot )$ is right continuous
(w.r.t. $x$)

\item For every $x\in \mathbb{R}$, $F_{\centerdot }(x)$ is decreasing and
left continuous (w.r.t. $m$).
\end{itemize}
\end{definition}

From Lemmas \ref{lemma2} and \ref{L3} we immediately deduce:

\begin{proposition}
\label{feasible}Let $\left\{ F_{m}\right\} _{m\in \mathbb{R}}$ be a feasible
family. Then the associated family $\left\{ \mathcal{A}^{m}\right\} _{m\in 
\mathbb{R}}$ is monotone decreasing and left continuous and each set $%
\mathcal{A}^{m}$ is convex and $\sigma (\mathcal{P},C_{b})-$closed. The
associated map $\Phi :\mathcal{P}\rightarrow \mathbb{R}\cup \{+\infty \}$ is
well defined, (Mon), (Qco) and $\sigma (\mathcal{P},C_{b})-$lsc
\end{proposition}

\begin{remark}
Let $\left\{ F_{m}\right\} _{m\in \mathbb{R}}$ be a feasible family. If
there exists an $\overline{m}$ such that $\lim_{x\rightarrow +\infty }F_{%
\overline{m}}(x)<1$ then $\lim_{x\rightarrow +\infty }F_{m}(x)<1$ for every $%
m\geq \overline{m}$ and then $\mathcal{A}^{m}=\emptyset $ for every $m\geq 
\overline{m}$. Obviously if an acceptability set is empty then it does not
contribute to the computation of the risk measure defined in (\ref{phi1}).
For this reason we will always consider w.l.o.g. a class $\left\{
F_{m}\right\} _{m\in \mathbb{R}}$ such that $\lim_{x\rightarrow +\infty
}F_{m}(x)=1$ for every $m$.
\end{remark}

\subsection{Examples}

As explained in the introduction, we define a family of risk measures
employing a Probability/Loss function $\Lambda $. Fix the \emph{right
continuous} function $\Lambda :\mathbb{R}\rightarrow \lbrack 0,1]$ and
define the family $\left\{ F_{m}\right\} _{m\in \mathbb{R}}$ of functions $%
F_{m}:\mathbb{R}\rightarrow \lbrack 0,1]$ by 
\begin{equation}
F_{m}(x):=\Lambda (x)\mathbf{1}_{(-\infty ,m)}(x)+\mathbf{1}_{[m,+\infty
)}(x).  \label{FF}
\end{equation}%
It is easy to check that if $\sup_{x\in \mathbb{R}}\Lambda (x)<1$ then the
family $\left\{ F_{m}\right\} _{m\in \mathbb{R}}$ is feasible and therefore,
by Proposition \ref{feasible}, the associated map $\Phi :\mathcal{P}%
\rightarrow \mathbb{R}\cup \{+\infty \}$ is well defined, (Mon), (Qco) and $%
\sigma (\mathcal{P},C_{b})-$lsc

\begin{example}
When $\sup_{x\in R}\Lambda (x)=1$, $\Phi $ may take the value $-\infty $.
The extreme case is when, in the definition of the family (\ref{FF}), the
function $\Lambda $ is equal to the constant one, $\Lambda (x)=1,$ and so: $%
\mathcal{A}^{m}=\mathcal{P}$ for all $m$ and $\Phi =-\infty .$
\end{example}

\begin{example}
\label{ex11}\textbf{Worst case risk measure: }$\Lambda (x)=0$\textbf{.}

Take in the definition of the family (\ref{FF}) the function $\Lambda $ to
be equal to the constant zero$:\Lambda (x)=0.$ Then: 
\begin{eqnarray*}
F_{m}(x) &:&=\mathbf{1}_{[m,+\infty )}(x) \\
\mathcal{A}^{m} &:&=\left\{ Q\in \mathcal{P}\mid F_{Q}\leq F_{m}\right\}
=\left\{ Q\in \mathcal{P}\mid \delta _{m} \preccurlyeq Q\right\} \\
\Phi _{w}(P) &:&=-\sup \left\{ m\mid P\in \mathcal{A}^{m}\right\} =-\sup
\left\{ m\mid \delta _{m}\preccurlyeq P \right\} \\
&=&-\sup \left\{ x\in \mathbb{R}\mid F_{P}(x)=0\right\}
\end{eqnarray*}%
so that, if $X\in L^{0}$ has distribution function $P_{X}$, 
\begin{equation*}
\Phi _{w}(P_{X})=-\sup \left\{ m\in \mathbb{R}\mid \delta _{m}\preccurlyeq
P_X\right\} =-ess\inf (X):=\rho _{w}(X)
\end{equation*}%
coincide with the worst case risk measure $\rho _{w}$. As the family $%
\{F_{m}\}$ is feasible, $\Phi _{w}:\mathcal{P}(\mathbb{R})\rightarrow 
\mathbb{R}\cup \{+\infty \}$ is (Mon), (Qco) and $\sigma (\mathcal{P}%
,C_{b})- $lsc In addition, it also satisfies (TrI).

Even though $\rho _{w}:L^{0}$ $\rightarrow \mathbb{R}\cup \{\infty \}$ is
convex, as a map defined on random variables, the corresponding $\Phi _{w}:%
\mathcal{P}\rightarrow \mathbb{R}\cup \{\infty \}$, as a map defined on
distribution functions, is not convex, but it is quasi-convex and concave.
Indeed, let $P\in \mathcal{P}$ and, since $F_{P}\geq 0,$ we set: 
\begin{equation*}
-\Phi _{w}(P)=\inf (F_{P}):=\sup \left\{ x\in \mathbb{R}:F_{P}(x)=0\right\} .
\end{equation*}%
If $F_{1}$, $F_{2}$ are two distribution functions corresponding to $P_{1}$, 
$P_{2}\in \mathcal{P}$ then for all $\lambda \in (0,1)$ we have: 
\begin{equation*}
\inf (\lambda F_{1}+(1-\lambda )F_{2})=\min (\inf (F_{1}),\inf (F_{2}))\leq
\lambda \inf (F_{1})+(1-\lambda )\inf (F_{2})
\end{equation*}%
and therefore, for all $\lambda \in \lbrack 0,1]$%
\begin{equation*}
\min (\inf (F_{1}),\inf (F_{2}))\leq \inf (\lambda F_{1}+(1-\lambda
)F_{2})\leq \lambda \inf (F_{1})+(1-\lambda )\inf (F_{2}).
\end{equation*}
\end{example}

\begin{example}
\label{ex12}\textbf{Value at Risk }$V@R_{\lambda }$: $\Lambda (x):=\lambda
\in (0,1).$

Take in the definition of the family (\ref{FF}) the function $\Lambda $ to
be equal to the constant $\lambda ,$ $\Lambda (x)=\lambda \in (0,1).$ Then 
\begin{eqnarray*}
F_{m}(x) &:&=\lambda \mathbf{1}_{(-\infty ,m)}(x)+\mathbf{1}_{[m,+\infty
)}(x) \\
\mathcal{A}^{m} &:&=\left\{ Q\in \mathcal{P}\mid F_{Q}\leq F_{m}\right\} \\
\Phi _{V@R_{\lambda }}(P) &:&=-\sup \left\{ m\in \mathbb{R}\mid P\in 
\mathcal{A}^{m}\right\}
\end{eqnarray*}%
If the random variable $X\in L^{0}$ has distribution function $P_{X}$ and $%
q_{X}^{+}(\lambda )=\sup \left\{ x\in \mathbb{R}\mid \mathbb{P}(X\leq x)\leq
\lambda \right\} $ is the right continuous inverse of $P_{X}$ then 
\begin{eqnarray*}
\Phi _{V@R_{\lambda }}(P_{X}) &=&-\sup \left\{ m\mid P_{X}\in \mathcal{A}%
^{m}\right\} \\
&=&-\sup \left\{ m\mid \mathbb{P}(X\leq x)\leq \lambda \;\forall x<m\right\}
\\
&=&-\sup \left\{ m\mid \mathbb{P}(X\leq m)\leq \lambda \right\} \\
&=&-q_{X}^{+}(\lambda ):=V@R_{\lambda }(X)
\end{eqnarray*}%
coincide with the Value At Risk of level $\lambda \in (0,1)$. As the family $%
\{F_{m}\}$ is feasible, $\Phi _{V@R_{\lambda }}:\mathcal{P}\rightarrow 
\mathbb{R}\cup \{+\infty \}$ is (Mon), (Qco), $\sigma (\mathcal{P},C_{b})$%
-lsc In addition, it also satisfies (TrI).

As well known, $V@R_{\lambda }:L^{0}$ $\rightarrow \mathbb{R}\cup \{\infty
\} $ is not quasi-convex, as a map defined on random variables, even though
the corresponding $\Phi _{V@R_{\lambda }}:\mathcal{P}\rightarrow \mathbb{R}%
\cup \{\infty \}$, as a map defined on distribution functions, is
quasi-convex (see \cite{DK10} for a discussion on this issue).
\end{example}


\begin{example}
Fix the family $\left\{ \Lambda _{m}\right\} _{m\in \mathbb{R}}$ of
functions $\Lambda _{m}:\mathbb{R}\rightarrow \lbrack 0,1]$ such that for
every $m\in \mathbb{R}$, $\Lambda _{m}(\centerdot )$ is right continuous
(w.r.t. $x$) and for every $x\in \mathbb{R}$, $\Lambda _{\centerdot }(x)$ is
decreasing and left continuous (w.r.t. $m$). Define the family $\left\{
F_{m}\right\} _{m\in \mathbb{R}}$ of functions $F_{m}:\mathbb{R}\rightarrow
\lbrack 0,1]$ by 
\begin{equation}
F_{m}(x):=\Lambda _{m}(x)\mathbf{1}_{(-\infty ,m)}(x)+\mathbf{1}_{[m,+\infty
)}(x).  \label{feas}
\end{equation}%
It is easy to check that if $\sup_{x\in \mathbb{R}}\Lambda _{m_{0}}(x)<1$,
for some $m_{0}\in \mathbb{R}$, then the family $\left\{ F_{m}\right\}
_{m\in \mathbb{R}}$ is feasible and therefore the associated map $\Phi :%
\mathcal{P}\rightarrow \mathbb{R}\cup \{+\infty \}$ is well defined, (Mon),
(Qco), $\sigma (\mathcal{P},C_{b})$-lsc
\end{example}

\section{On the $\Lambda V@R$}


We now propose a generalization of the $V@R_{\lambda }$ which appears useful
for possible application whenever an agent is facing some ambiguity on the
parameter $\lambda $, namely $\lambda $ is given by some uncertain value in
a confidence interval $[\lambda ^{m},\lambda ^{M}]$, with $0\leq \lambda
^{m}\leq \lambda ^{M}\leq 1$. The $V@R_{\lambda }$ corresponds to case $%
\lambda ^{m}=\lambda ^{M}$ and one typical value is $\lambda ^{M}=0,05$.

We will distinguish two possible classes of agents:

\paragraph{Risk prudent Agents}

Fix the \emph{increasing} right continuous function $\Lambda :\mathbb{R}%
\rightarrow \lbrack 0,1]$, choose as in (\ref{FF}) 
\begin{equation*}
F_{m}(x)=\Lambda (x)\mathbf{1}_{(-\infty ,m)}(x)+\mathbf{1}_{[m,+\infty )}(x)
\end{equation*}%
and set $\lambda ^{m}:=\inf \Lambda \geq 0$, $\lambda ^{M}:=\sup \Lambda
\leq 1$. As the function $\Lambda $ is increasing, we are assigning to a
lower loss a lower probability. In particular given two possible choices $%
\Lambda _{1},\Lambda _{2}$ for two different agents, the condition $\Lambda
_{1}\leq \Lambda _{2}$ means that the agent 1 is more risk prudent than
agent 2. \newline
Set, as in (\ref{am}), $\mathcal{A}^{m}=\left\{ Q\in \mathcal{P}\mid
F_{Q}\leq F_{m}\right\} $ and define as in (\ref{phi1}) 
\begin{equation*}
\Lambda V@R(P):=-\sup \left\{ m\in \mathbb{R}\mid P\in \mathcal{A}%
^{m}\right\} .
\end{equation*}%
Thus, in case of a random variable $X$ 
\begin{equation*}
\Lambda V@R(P_{X}):=-\sup \left\{ m\in \mathbb{R}\mid \mathbb{P}(X\leq
x)\leq \Lambda (x),\;\forall x\leq m\right\} .
\end{equation*}%
In particular it can be rewritten as 
\begin{equation*}
\Lambda V@R(P_{X})=-\inf \left\{ x\in \mathbb{R}\mid \mathbb{P}(X\leq
x)>\Lambda (x)\right\} .
\end{equation*}%
If both $F_{X}$ and $\Lambda $ are continuous $\Lambda V@R$ corresponds to
the smallest intersection between the two curves.

In this section, we assume that 
\begin{equation*}
\lambda ^{M}<1.
\end{equation*}%
Besides its obvious financial motivation, this request implies that the
corresponding family $F_{m}$ is feasible and so $\Lambda V@R(P)>-\infty $
for all $P\in \mathcal{P}$.

The feasibility of the family $\left\{ F_{m}\right\} $ implies that the $%
\Lambda V@R:\mathcal{P}\rightarrow \mathbb{R\cup }\left\{ \infty \right\} $
is well defined, (Mon), (QCo) and (CfA) (or equivalently $\sigma (\mathcal{P}%
,C_{b})$-lsc) map.

\begin{example}
One possible simple choice of the function $\Lambda $ is \ represented by
the step function:%
\begin{equation*}
\Lambda (x)=\lambda ^{m}\mathbf{1}_{(-\infty ,\bar{x})}(x)+\lambda ^{M}%
\mathbf{1}_{[\bar{x},+\infty )}(x)
\end{equation*}%
The idea is that with a probability of $\lambda ^{M}$ we are accepting to
loose at most $\bar{x}$. In this case we observe that: 
\begin{equation*}
\Lambda V@R(P)=\left\{ 
\begin{array}{cc}
V@R_{\lambda ^{M}}(P) & \text{if }V@R_{\lambda ^{m}}(P)\leq -\bar{x} \\ 
V@R_{\lambda ^{m}}(P) & \text{if }V@R_{\lambda ^{m}}(P)>-\bar{x}.%
\end{array}%
\right. 
\end{equation*}%
Even though the $\Lambda V@R$ is continuous from above (Proposition \ref%
{feasible} and \ref{CFB}), it may not be continuous from below, as this
example shows. For instance take $\bar{x}=0$ and $P_{X_{n}}$ induced by a
sequence of uniformly distributed random variables $X_{n}\sim U\left[
-\lambda ^{m}-\frac{1}{n},1-\lambda ^{m}-\frac{1}{n}\right] $. We have $%
P_{X_{n}}\uparrow P_{U[-\lambda ^{m},1-\lambda ^{m}]}$ but $\Lambda
V@R(P_{X_{n}})=-\frac{1}{n}$ for every $n$ and $\Lambda V@R(P_{U[-\lambda
^{m},1-\lambda ^{m}]})=\lambda ^{M}-\lambda ^{m}$.
\end{example}

\begin{remark}
(i) If $\lambda ^{m}=0$ the domain of $\Lambda V@R(P)$ is not the entire
convex set $\mathcal{P}$. We have two possible cases

\begin{itemize}
\item $supp(\Lambda )=[x^{\ast },+\infty )$: in this case $\Lambda
V@R(P)=-\inf supp(F_{P})$ for every $P\in \mathcal{P}$ such that $%
supp(F_{P})\supseteq supp(\Lambda )$.

\item $supp(\Lambda )=(-\infty ,+\infty )$: in this case 
\begin{eqnarray*}
\Lambda V@R(P)=+\infty &&\text{for all }P\text{ such that }%
\lim_{x\rightarrow -\infty }\frac{F_{P}(x)}{\Lambda (x)}>1 \\
\Lambda V@R(P)<+\infty &&\text{for all }P\text{ such that }%
\lim_{x\rightarrow -\infty }\frac{F_{P}(x)}{\Lambda (x)}<1
\end{eqnarray*}
\end{itemize}

In the case $\lim_{x\rightarrow -\infty }\frac{F_{P}(x)}{\Lambda (x)}=1$
both the previous behaviors might occur.

\noindent (ii) In case that $\lambda ^{m}>0$ then $\Lambda V@R(P)<+\infty $
for all $P\in \mathcal{P}$, so that $\Lambda V@R$ is finite valued.
\end{remark}

\bigskip

We can prove a further structural property which is the counterpart of (TrI)
for the $\Lambda V@R$. Let $\alpha \in \mathbb{R}$ any cash amount 
\begin{eqnarray*}
\Lambda V@R(P_{X+\alpha }) &=&-\sup \left\{ m\mid \mathbb{P}(X+\alpha \leq
x)\leq \Lambda (x),\;\forall x\leq m\right\} \\
&=&-\sup \left\{ m\mid \mathbb{P}(X\leq x-\alpha )\leq \Lambda (x),\;\forall
x\leq m\right\} \\
&=&-\sup \left\{ m\mid \mathbb{P}(X\leq y)\leq \Lambda (y+\alpha ),\;\forall
y\leq m-\alpha \right\} \\
&=&-\sup \left\{ m+\alpha \mid \mathbb{P}(X\leq y)\leq \Lambda (y+\alpha
),\;\forall y\leq m\right\} \\
&=&\Lambda ^{\alpha }V@R(P_{X})-\alpha
\end{eqnarray*}%
where $\Lambda ^{\alpha }(x)=\Lambda (x+\alpha )$. We may conclude that if
we add a sure positive (resp. negative) amount $\alpha $ to a risky position 
$X$ then the risk decreases (resp. increases) of the value $-\alpha $,
constrained to a lower (resp. higher) level of risk prudence described by $%
\Lambda ^{\alpha }\geq \Lambda $ (resp. $\Lambda ^{\alpha }\leq \Lambda $).
For an arbitrary $P\in \mathcal{P}$ this property can be written as 
\begin{equation*}
\Lambda V@R(T_{\alpha }P)=\Lambda ^{\alpha }V@R(P)-\alpha ,\quad \forall
\,\alpha \in \mathbb{R},
\end{equation*}%
where $T_{\alpha }P(-\infty ,x]=P(-\infty ,x-\alpha ]$.

\paragraph{Risk Seeking Agents}

Fix the \emph{decreasing} right continuous function $\Lambda :\mathbb{R}%
\rightarrow \lbrack 0,1]$, with $\inf \Lambda <1$. Similarly as above, we
define 
\begin{equation*}
F_{m}(x)=\Lambda (x)\mathbf{1}_{(-\infty ,m)}(x)+\mathbf{1}_{[m,+\infty )}(x)
\end{equation*}%
and the (Mon), (QCo) and (CfA) map 
\begin{equation*}
\Lambda V@R(P):=-\sup \left\{ m\in \mathbb{R}\mid F_{P}\leq F_{m}\right\}
=-\sup \left\{ m\in \mathbb{R}\mid \mathbb{P}(X\leq m)\leq \Lambda
(m)\right\} .
\end{equation*}%
In this case, for eventual huge losses we are allowing the highest level of
probability. As in the previous example let $\alpha \in \mathbb{R}$ and
notice that 
\begin{equation*}
\Lambda V@R(P_{X+\alpha })=\Lambda ^{\alpha }V@R(P_{X})-\alpha .
\end{equation*}%
where $\Lambda ^{\alpha }(x)=\Lambda (x+\alpha )$. The property is exactly
the same as in the former example but here the interpretation is slightly
different. If we add a sure positive (resp. negative) amount $\alpha $ to a
risky position $X$ then the risk decreases (resp. increases) of the value $%
-\alpha $, constrained to a lower (resp. higher) level of risk seeking since 
$\Lambda ^{\alpha }\leq \Lambda $ (resp. $\Lambda ^{\alpha }\geq \Lambda $).

\begin{remark}
\label{remarkDecr}For a decreasing $\Lambda ,$ there is a simpler
formulation - which will be used in Section 5.3 - of the $\Lambda V@R$ that
is obtained replacing in $F_{m}$ the function $\Lambda $ with the line $%
\Lambda (m)$ for all $x<m$. Let 
\begin{equation*}
\tilde{F}_{m}(x)=\Lambda (m)\mathbf{1}_{(-\infty ,m)}(x)+\mathbf{1}%
_{[m,+\infty )}(x).
\end{equation*}%
This family is of the type (\ref{feas}) and is feasible, provided the
function $\Lambda $ is continuous. For a decreasing $\Lambda ,$ it is
evident that 
\begin{equation*}
\Lambda V@R(P)=\Lambda \widetilde{V}@R(P):=-\sup \left\{ m\in \mathbb{R}\mid
F_{P}\leq \tilde{F}_{m}\right\} ,
\end{equation*}%
as the function $\Lambda $ lies above the line $\Lambda (m)$ for all $x\leq
m $.
\end{remark}


\section{Quasi-convex Duality}


In literature we also find several results about the dual representation of
law invariant risk measures. Kusuoka \cite{K01} contributed to the coherent
case, while Frittelli and Rosazza \cite{FR05} extended this result to the
convex case. Jouini, Schachermayer and Touzi (2006) \cite{JST06}, in the
convex case, and Svindland (2010) \cite{S10} in the quasi-convex case,
showed that every law invariant risk measure is already weakly lower
semicontinuous. Recently, Cerreia-Vioglio, Maccheroni, Marinacci and
Montrucchio (2010) \cite{CMMMa} provided a robust dual representation for
law invariant quasi-convex risk measures, which has been extended to the
dynamic case in \cite{FM09}.

In Sections 5.1 and 5.2 we will treat the general case of maps defined on $%
\mathcal{P}$, while in Section 5.3 we specialize these results to show the
dual representation of maps associated to feasible families.

\subsection{Reasons of the failure of the convex duality for Translation
Invariant maps on $\mathcal{P}$}

It is well known that the classical convex duality provided by the
Fenchel-Moreau theorem guarantees the representation of convex and lower
semicontinuous functions and therefore is very useful for the dual
representation of convex risk measures (see \cite{fr}). For any map $\Phi :%
\mathcal{P}\rightarrow \mathbb{R\cup }\left\{ \infty \right\} $ let $\Phi
^{\ast }$ be the convex conjugate: 
\begin{equation*}
\Phi ^{\ast }(f):=\sup_{Q\in \mathcal{P}}\left\{ \int fdQ-\Phi (Q)\right\} 
\text{, }f\in C_{b}.
\end{equation*}%
Applying the fact that $\mathcal{P}$ is a $\sigma (ca,C_{b})$ closed convex
subset of $ca$ one can easily check that the following version of
Fenchel-Moreau Theorem holds true for maps defined on $\mathcal{P}$.

\begin{proposition}[Fenchel-Moreau]
\label{FM}Suppose that $\Phi :\mathcal{P}\rightarrow \mathbb{R\cup }\left\{
\infty \right\} \mathbb{\ }$is $\sigma (\mathcal{P},C_{b})-$ lsc and convex.
If $Dom(\Phi ):=\left\{ Q\in \mathcal{P}\mid \Phi (Q)<+\infty \right\} \neq
\varnothing $ then $Dom(\Phi ^{\ast })\neq \varnothing $ and 
\begin{equation*}
\Phi (Q)=\sup_{f\in C_{b}}\left\{ \int fdQ-\Phi ^{\ast }(f)\right\} .
\end{equation*}
\end{proposition}

One trivial example of a proper $\sigma (\mathcal{P},C_{b})-$lsc and convex
map on $\mathcal{P}$ is given by $Q\rightarrow \int fdQ$, for some $f\in
C_{b}$. But this map does not satisfy the (TrI) property. Indeed, we show
that in the setting of risk measures defined on $\mathcal{P}$, weakly lower
semicontinuity and convexity are incompatible with translation invariance.

\begin{proposition}
\label{NO}For any map $\Phi :\mathcal{P}\rightarrow \mathbb{R\cup }\left\{
\infty \right\} $, if there exists a sequence $\left\{ Q_{n}\right\}
_{n}\subseteq \mathcal{P}$ such that \underline{$\lim $}$_{n}\Phi
(Q_{n})=-\infty $ then $Dom(\Phi ^{\ast })=\varnothing .$
\end{proposition}

\begin{proof}
For any $f\in C_{b}(\mathbb{R})$ 
\begin{equation*}
\Phi ^{\ast }(f)=\sup_{Q\in \mathcal{P}}\left\{ \int fdQ-\Phi (Q)\right\}
\geq \int fd(Q_{n})-\Phi (Q_{n})\geq \inf_{x\in \mathbb{R}}f(x)-\Phi (Q_{n}),
\end{equation*}%
which implies $\Phi ^{\ast }=+\infty $.
\end{proof}

\bigskip 

From Propositions (\ref{FM}) and (\ref{NO}) we immediately obtain:

\begin{corollary}
\label{cor}Let $\Phi :\mathcal{P}\rightarrow \mathbb{R\cup }\left\{ \infty
\right\} $ be $\sigma (\mathcal{P},C_{b})$-lsc, convex and not identically
equal to $+\infty $. Then $\Phi $ is not (TrI), is not cash sup additive
(i.e. it does not satisfy: $\Phi (T_{m}Q)\leq \Phi (Q)-m$ ) and \underline{$%
\lim $}$_{n}\Phi (\delta _{n})\neq -\infty $. In particular, the certainty
equivalent maps $\Phi _{f}$ defined in (\ref{ceq}) can not be convex, as
they are $\sigma (\mathcal{P},C_{b})$-lsc and $\Phi _{f}(\delta _{n})=-n$
\end{corollary}

\subsection{The dual representation}

As described in the Examples in Section 3, the $\Phi _{V@R_{\lambda }}$ and $%
\Phi _{w}$ are proper, $\sigma (ca,C_{b})-$lsc, quasi-convex (Mon) and (TrI)
maps $\Phi :\mathcal{P}\rightarrow \mathbb{R\cup }\left\{ \infty \right\} $.
Therefore, the negative result outlined in Corollary \ref{cor} for the
convex case can not be true in the quasi-convex setting.

We recall that the seminal contribution to quasi-convex duality comes from
the dual representation by Volle \cite{Volle}, which has been sharpened to a
complete quasiconvex duality by Cerreia-Vioglio et al. \cite{CMMMa} (case of
M-spaces), Cerreia-Vioglio \cite{CV} (preferences over menus) and Drapeau
and Kupper \cite{DK10} (for general topological vector spaces).

Here we replicate this result and provide the dual representation of a $%
\sigma (\mathcal{P},C_{b})$ lsc quasi-convex maps defined on the entire set $%
\mathcal{P}$. The main difference is that our map $\Phi $ is defined on a
convex subset of $ca$ and not a vector space (a similar result can be found
in \cite{DK10} for convex sets). But since $\mathcal{P}$ is $\sigma
(ca,C_{b})$-closed, the first part of the proof will match very closely the
one given by Volle. In order to achieve the dual representation of $\sigma (%
\mathcal{P},C_{b})$ lsc risk measures $\Phi :\mathcal{P}\rightarrow \mathbb{%
R\cup }\left\{ \infty \right\} $ we will impose the monotonicity assumption
of $\Phi $ and deduce that in the dual representation the supremum can be
restricted to the set 
\begin{equation*}
C_{b}^{-}=\left\{ f\in C_{b}\mid f\text{ is decreasing}\right\} .
\end{equation*}%
This is natural as the first order stochastic dominance implies (see Th.
2.70 \cite{FoSch})\ that

\begin{equation}
C_{b}^{-}=\left\{ f\in C_{b}\mid Q,P\in \mathcal{P}\text{ and }P\preccurlyeq
Q \Rightarrow \int fdQ\leq \int fdP\right\} .  \label{cb}
\end{equation}

Notice that differently from \cite{DK10} the following proposition does not
require the extension of the risk map to the entire space $ca(\mathbb{R})$.
Once the representation is obtained the uniqueness of the dual function is a
direct consequence of Theorem 2.19 in \cite{DK10} as explained by
Proposition \ref{unique}.

\begin{proposition}
\label{propvolleMon}(i) Any $\sigma (\mathcal{P},C_{b})-$lsc and
quasi-convex functional $\Phi :\mathcal{P}\rightarrow \mathbb{R}\cup \left\{
\infty \right\} $ can be represented as 
\begin{equation}
\Phi (P)=\sup_{f\in C_{b}}R\left( \int fdP,f\right)  \label{repr}
\end{equation}%
where $R:\mathbb{R}\times C_{b}\rightarrow \overline{\mathbb{R}}$ is defined
by 
\begin{equation}
R(t,f):=\inf_{Q\in \mathcal{P}}\left\{ \Phi (Q)\mid \int fdQ\geq t\right\} .
\label{122}
\end{equation}%
(ii) If in addition $\Phi $ is monotone then (\ref{repr}) holds with $C_{b}$
replaced by $C_{b}^{-}.$
\end{proposition}

\begin{proof}
We will use the fact that $\sigma (\mathcal{P},C_{b})$ is the relativization
of $\sigma (ca,C_{b})$ to the set $\mathcal{P}$. In particular the lower
level sets will be $\sigma (ca,C_{b})$-closed.

(i) By definition, for any $f\in C_{b}(\mathbb{R})$, $R\left( \int
fdP,f\right) \leq \Phi (P)$ and therefore 
\begin{equation*}
\sup_{f\in C_{b}}R\left( \int fdP,f\right) \leq \Phi (P),\quad P\in \mathcal{%
P}.
\end{equation*}%
Fix any $P\in \mathcal{P}$ and take $\varepsilon \in \mathbb{R}$ such that $%
\varepsilon >0$. Then $P$ does not belong to the $\sigma (ca,C_{b})$-closed
convex set 
\begin{equation*}
\mathcal{C}_{\varepsilon }:=\left\{ Q\in \mathcal{P}:\Phi (Q)\leq \Phi
(P)-\varepsilon \right\}
\end{equation*}
(if $\Phi (P)=+\infty $, replace the set $\mathcal{C}_{\varepsilon }$ with $%
\left\{ Q\in \mathcal{P}:\Phi (Q)\leq M\right\} ,$ for any $M$). By the Hahn
Banach theorem there exists a continuous linear functional that strongly
separates $P$ and $\mathcal{C}_{\varepsilon }$, i.e. there exists $\alpha
\in \mathbb{R}$ and $f_{\varepsilon }\in C_{b}$ such that 
\begin{equation}
\int f_{\varepsilon }dP>\alpha >\int f_{\varepsilon }dQ\quad \text{ for all }%
Q\in \mathcal{C}_{\varepsilon }\text{.}  \label{100}
\end{equation}%
Hence:%
\begin{equation}
\left\{ Q\in \mathcal{P}:\int f_{\varepsilon }dP\leq \int f_{\varepsilon
}dQ\right\} \subseteq (\mathcal{C}_{\varepsilon })^{C}=\left\{ Q\in \mathcal{%
P}:\Phi (Q)>\Phi (P)-\varepsilon \right\}  \label{finitecase}
\end{equation}%
and 
\begin{eqnarray}
\Phi (P) &\geq &\sup_{f\in C_{b}}R\left( \int fdP,f\right) \geq R\left( \int
f_{\varepsilon }dP,f_{\varepsilon }\right)  \notag \\
&=&\inf \left\{ \Phi (Q)\mid Q\in \mathcal{P}\text{ such that }\int
f_{\varepsilon }dP\leq \int f_{\varepsilon }dQ\right\}  \notag \\
&\geq &\inf \left\{ \Phi (Q)\mid Q\in \mathcal{P}\text{ satisfying }\Phi
(Q)>\Phi (P)-\varepsilon \right\} \geq \Phi (P)-\varepsilon .  \label{105}
\end{eqnarray}%
(ii) We furthermore assume that $\Phi $ is monotone. As shown in (i), for
every $\varepsilon >0$ we find $f_{\varepsilon }$ such that (\ref{100})
holds true. We claim that there exists $g_{\varepsilon }\in C_{b}^{-}$
satisfying: 
\begin{equation}
\int g_{\varepsilon }dP>\alpha >\int g_{\varepsilon }dQ\quad \text{ for all }%
Q\in \mathcal{C}_{\varepsilon }.
\end{equation}%
and then the above argument (in equations (\ref{100})-(\ref{105})) implies
the thesis.

We define the decreasing function 
\begin{equation*}
g_{\varepsilon }(x)=:\sup_{y\geq x}f_{\varepsilon }(y)\in C_{b}^{-}.
\end{equation*}

\emph{First case:} suppose that $g_{\varepsilon }(x)=\sup_{x\in \mathbb{R}%
}f_{\varepsilon }(x)=:s$. In this case there exists a sequence of $%
\{x_{n}\}_{n\in \mathbb{N}}\subseteq \mathbb{R}$ such that $x_{n}\rightarrow
+\infty $ and $f_{\varepsilon }(x_{n})\rightarrow s$, as $n\rightarrow
\infty $. Define 
\begin{equation*}
g_{n}(x)=s\mathbf{1}_{(-\infty ,x_{n}]}+f_{\varepsilon }(x)\mathbf{1}%
_{(x_{n},+\infty )}
\end{equation*}%
and notice that $s\geq g_{n}\geq f_{\varepsilon }$ and $g_{n}\uparrow s$.
For any $Q\in \mathcal{C}_{\varepsilon }$ we consider $Q_{n}$ defined by $%
F_{Q_{n}}(x)=F_{Q}(x)\mathbf{1}_{[x_{n},+\infty )}$. Since $Q\preccurlyeq
Q_n $, monotonicity of $\Phi $ implies $Q_{n}\in \mathcal{C}_{\varepsilon }$%
. Notice that 
\begin{equation}
\int g_{n}dQ-\int f_{\varepsilon }dQ_{n}=(s-f_{\varepsilon
}(x_{n}))Q(-\infty ,x_{n}]\overset{n\rightarrow +\infty }{\longrightarrow }0,%
\text{ as }n\rightarrow \infty .  \label{99}
\end{equation}%
From equation (\ref{100}) we have%
\begin{equation}
s\geq \int f_{\varepsilon }dP>\alpha >\int f_{\varepsilon }dQ_{n}\quad \text{
for all }n\in \mathbb{N}.  \label{101}
\end{equation}%
Letting $\delta =s-\alpha >0$ we obtain $s>\int f_{\varepsilon }dQ_{n}+\frac{%
\delta }{2}$. From (\ref{99}), there exists $\overline{n}\in \mathbb{N}$
such that $0\leq \int g_{n}dQ-\int f_{\varepsilon }dQ_{n}<\frac{\delta }{4}$
for every $n\geq \overline{n}.$ Therefore $\forall \,n\geq \overline{n}$ 
\begin{equation*}
s>\int f_{\varepsilon }dQ_{n}+\frac{\delta }{2}>\int g_{n}dQ-\frac{\delta }{4%
}+\frac{\delta }{2}=\int g_{n}dQ+\frac{\delta }{4}
\end{equation*}%
and this leads to a contradiction since $g_{n}\uparrow s$. So the first case
is excluded.

\emph{Second case:} suppose that $g_{\varepsilon }(x)<s$ for any $x>%
\overline{x}$. As the function $g_{\varepsilon }\in C_{b}^{-}$ is
decreasing, there will exists at most a countable sequence of intervals $%
\left\{ A_{n}\right\} _{n\geq 0}$ on which $g_{\varepsilon }$ is constant.
Set $A_{0}=(-\infty ,b_{0}),$ $A_{n}=[a_{n},b_{n})\subset \mathbb{R}$ for $%
n\geq 1$. W.l.o.g. we suppose that $A_{n}\cap A_{m}=\emptyset $ for all $%
n\neq m$ (else, we paste together the sets) and $a_{n}<a_{n+1}$ for every $%
n\geq 1$. We stress that $f_{\varepsilon }(x)=g_{\varepsilon }(x)$ on $%
D=:\bigcap_{n\geq 0}A_{n}^{C}$. For every $Q\in \mathcal{C}_{\varepsilon }$
we define the probability $\overline{Q}$ by its distribution function as 
\begin{equation*}
F_{\overline{Q}}(x)=F_{Q}(x)\mathbf{1}_{D}+\sum_{n\geq 1}F_{Q}(a_{n})\mathbf{%
1}_{[a_{n},b_{n})}.
\end{equation*}%
As before, $Q\preccurlyeq \overline{Q}$ and monotonicity of $\Phi $ implies $%
\overline{Q}\in \mathcal{C}_{\varepsilon }$. Moreover 
\begin{equation*}
\int g_{\varepsilon }dQ=\int_{D}f_{\varepsilon }dQ+f_{\varepsilon
}(b_{0})Q(A_{0})+\sum_{n\geq 1}f_{\varepsilon }(a_{n})Q(A_{n})=\int
f_{\varepsilon }d\overline{Q}.
\end{equation*}%
From $g_{\varepsilon }\geq f_{\varepsilon }$ and equation (\ref{100}) we
deduce%
\begin{equation*}
\int g_{\varepsilon }dP\geq \int f_{\varepsilon }dP>\alpha >\int
f_{\varepsilon }d\overline{Q}=\int g_{\varepsilon }dQ\quad \text{ for all }%
Q\in \mathcal{C}_{\varepsilon }.
\end{equation*}
\end{proof}

We reformulate the Proposition \ref{propvolleMon} and provide two dual
representation of $\sigma (\mathcal{P}(\mathbb{R}),C_{b})$-lsc Risk Measure $%
\Phi :\mathcal{P}(\mathbb{R})\rightarrow \mathbb{R}\cup \left\{ \infty
\right\} $ in terms of a supremum over a class of probabilistic scenarios.
Let

\begin{equation*}
\mathcal{P}_{c}(\mathbb{R})=\left\{ Q\in \mathcal{P}(\mathbb{R})\mid F_{Q}%
\text{ is continuous}\right\}.
\end{equation*}

\begin{proposition}
\label{propvolle copy(1)}Any $\sigma (\mathcal{P}(\mathbb{R}),C_{b})$-lsc
Risk Measure $\Phi :\mathcal{P}(\mathbb{R})\rightarrow \mathbb{R}\cup
\left\{ \infty \right\} $ can be represented as 
\begin{equation*}
\Phi (P)=\sup_{Q\in \mathcal{P}_{c}(\mathbb{R})}R\left( -\int
F_{Q}dP,-F_{Q}\right) .
\end{equation*}
\end{proposition}

\begin{proof}
Notice that for every $f\in C_{b}^{-}$ which is constant we have $R(\int
fdP,f)=\inf_{Q\in \mathcal{P}}\Phi (Q)$. Therefore we may assume w.l.o.g.
that $f\in C_{b}^{-}$ is not constant. Then $g:=\frac{f-f(+\infty )}{%
f(-\infty )-f(+\infty )}\in C_{b}^{-}$, $\inf g=0,$ $\sup g=1$, and so: $%
g\in \left\{ -F_{Q}\mid Q\in \mathcal{P}_{c}(\mathbb{R})\right\} $. In
addition, since $\int fdQ\geq \int fdP$ iff $\int gdQ\geq \int gdP$ we
obtain from (\ref{repr}) and ii) of Proposition \ref{propvolleMon} 
\begin{equation*}
\Phi (P)=\sup_{f\in C_{b}^{-}}R\left( \int fdP,f\right) =\sup_{Q\in \mathcal{%
P}_{c}(\mathbb{R})}R\left( -\int F_{Q}dP,-F_{Q}\right) .
\end{equation*}
\end{proof}

Finally we state the dual representations for Risk Measures expressed either
in terms of the dual function $R$ as used by \cite{CMMMa}, or considering
the left continuous version of $R$ (see Lemma \ref{Lleft}) in the
formulation proposed by \cite{DK10}. If $R:\mathbb{R}\times C_{b}(\mathbb{R}%
)\rightarrow \overline{\mathbb{R}}$, the left continuous version of $R(\cdot
,f)$ is defined by: 
\begin{equation}
R^{-}(t,f):=\sup \left\{ R(s,f)\mid s<t\right\} .  \label{16}
\end{equation}

\begin{proposition}
\label{propvolle}Any $\sigma (\mathcal{P}(\mathbb{R}),C_{b})$-lsc Risk
Measure $\Phi :\mathcal{P}(\mathbb{R})\rightarrow \mathbb{R}\cup \left\{
\infty \right\} $ can be represented as 
\begin{equation}
\Phi (P)=\sup_{f\in C_{b}^{-}}R\left( \int fdP,f\right) =\sup_{f\in
C_{b}^{-}}R^{-}\left( \int fdP,f\right) .  \label{11}
\end{equation}%
The function $R^{-}(t,f)$ defined in (\ref{16}) can be written as 
\begin{equation}
R^{-}(t,f)=\inf \left\{ m\in \mathbb{R}\mid \gamma (m,f)\geq t\right\} ,
\label{1111}
\end{equation}%
where $\gamma :\mathbb{R}\times C_{b}(\mathbb{R})\rightarrow \overline{%
\mathbb{R}}$ is given by: 
\begin{equation}
\gamma (m,f):=\sup_{Q\in \mathcal{P}}\left\{ \int fdQ\mid \Phi (Q)\leq
m\right\} \text{, }m\in \mathbb{R}.  \label{123}
\end{equation}
\end{proposition}

\begin{proof}
Notice that $R(\cdot ,f)$ is increasing and $R\left( t,f\right) \geq
R^{-}\left( t,f\right) .$ If $f\in C_{b}^{-}$ then $P\preccurlyeq
Q\Rightarrow \int fdQ\leq \int fdP$. Therefore,%
\begin{equation*}
R^{-}\left( \int fdP,f\right) :=\sup_{s<\int fdP}R(s,f)\geq
\lim_{P_{n}\downarrow P}R(\int fdP_{n},f).
\end{equation*}%
From Proposition \ref{propvolleMon} (ii) we obtain:%
\begin{align*}
\Phi (P)& =\sup_{f\in C_{b}^{-}}R\left( \int fdP,f\right) \geq \sup_{f\in
C_{b}^{-}}R^{-}\left( \int fdP,f\right) \geq \sup_{f\in
C_{b}^{-}}\lim_{P_{n}\downarrow P}R(\int fdP_{n},f) \\
& =\lim_{P_{n}\downarrow P}\sup_{f\in C_{b}^{-}}R(\int
fdP_{n},f)=\lim_{P_{n}\downarrow P}\Phi (P_{n})=\Phi (P).
\end{align*}%
by (CfA). This proves (\ref{11}). The second statement follows from the
Lemma \ref{Lleft}.
\end{proof}

The following Lemma shows that the left continuous version of $R$ is the
left inverse of the function $\gamma $ as defined in \ref{123} (for the
definition and the properties of the left inverse we refer to \cite{FoSch}
Section A.3).

\begin{lemma}
\label{Lleft}Let $\Phi $ be any map $\Phi :\mathcal{P}(\mathbb{R}%
)\rightarrow \mathbb{R}\cup \left\{ \infty \right\} $ and $R:\mathbb{R}%
\times C_{b}(\mathbb{R})\rightarrow \overline{\mathbb{R}}$ be defined in (%
\ref{122}). The left continuous version of $R(\cdot ,f)$ can be written as: 
\begin{equation}
R^{-}(t,f):=\sup \left\{ R(s,f)\mid s<t\right\} =\inf \left\{ m\in \mathbb{R}%
\mid \gamma (m,f)\geq t\right\} ,  \label{S}
\end{equation}%
where $\gamma :\mathbb{R}\times C_{b}(\mathbb{R})\rightarrow \overline{%
\mathbb{R}}$ is given in (\ref{123}).
\end{lemma}

\begin{proof}
Let the RHS of equation (\ref{S}) be denoted by 
\begin{equation*}
S(t,f):=\inf \left\{ m\in \mathbb{R}\mid \gamma (m,f)\geq t\right\} ,\text{ }%
(t,f)\in \mathbb{R}\times C_{b}(\mathbb{R}),
\end{equation*}%
and note that $S(\cdot ,f)$ is the left inverse of the increasing function $%
\gamma (\cdot ,f)$ and therefore $S(\cdot ,f)$ is left continuous.\newline
Step I. To prove that $R^{-}(t,f)\geq S(t,f)$ it is sufficient to show that
for all $s<t$ we have:%
\begin{equation}
R(s,f)\geq S(s,f),  \label{R>=S}
\end{equation}%
Indeed, if (\ref{R>=S}) is true%
\begin{equation*}
R^{-}(t,f)=\sup_{s<t}R(s,f)\geq \sup_{s<t}S(s,f)=S(t,f),
\end{equation*}%
as both $R^{-}$ and $S$ are left continuous in the first argument.\newline
Writing explicitly the inequality (\ref{R>=S})%
\begin{equation*}
\inf_{Q\in \mathcal{P}}\left\{ \Phi (Q)\mid \int fdQ\geq s\right\} \geq \inf
\left\{ m\in \mathbb{R}\mid \gamma (m,f)\geq s\right\}
\end{equation*}%
and letting $Q\in \mathcal{P}$ satisfying $\int fdQ\geq s$, we see that it
is sufficient to show the existence of $m\in \mathbb{R}$ such that $\gamma
(m,f)\geq s$ and $m\leq \Phi (Q)$. If $\Phi (Q)=-\infty $ then $\gamma
(m,f)\geq s$ for any $m$ and therefore $S(s,f)=R(s,f)=-\infty $.

Suppose now that $\infty >\Phi (Q)>-\infty $ and define $m:=\Phi (Q).$ As $%
\int fdQ\geq s$ we have:%
\begin{equation*}
\gamma (m,f):=\sup_{Q\in \mathcal{P}}\left\{ \int fdQ\mid \Phi (Q)\leq
m\right\} \geq s
\end{equation*}%
Then $m\in \mathbb{R}$ satisfies the required conditions.

Step II : To obtain $R^{-}(t,f):=\sup_{s<t}R(s,f)\leq S(t,f)$ it is
sufficient to prove that, for all $s<t,$ $R(s,f)\leq S(t,f)$, that is%
\begin{equation}
\inf_{Q\in \mathcal{P}}\left\{ \Phi (Q)\mid \int fdQ\geq s\right\} \leq \inf
\left\{ m\in \mathbb{R}\mid \gamma (m,f)\geq t\right\} .  \label{R<=S}
\end{equation}

Fix any $s<t$\ and consider any $m\in \mathbb{R}$ such that $\gamma
(m,f)\geq t$. By the definition of $\gamma $, for all $\varepsilon >0$ there
exists $Q_{\varepsilon }\in \mathcal{P}$ such that $\Phi (Q_{\varepsilon
})\leq m$ and $\int fdQ_{\varepsilon }>t-\varepsilon .$ Take $\varepsilon $
such that $0<\varepsilon <t-s$. Then $\int fdQ_{\varepsilon }\geq s$ and $%
\Phi (Q_{\varepsilon })\leq m$ and (\ref{R<=S}) follows.
\end{proof}

\paragraph{Complete duality}

The complete duality in the class of quasi-convex monotone maps on vector
spaces was first obtained by \cite{CMMMb}. The following proposition is
based on the complete duality proved in \cite{DK10} for maps defined on
convex sets and therefore the results in \cite{DK10} apply very easily in
our setting. In order to obtain the uniqueness of the dual function in the
representation (\ref{11}) we need to introduce the opportune class $\mathcal{%
R}^{\max }$. Recall that $\mathcal{P}(\mathbb{R})$ spans the space of
countably additive signed measures on $\mathbb{R}$, namely $ca(\mathbb{R})$
and that the first stochastic order corresponds to the cone 
\begin{equation*}
\mathcal{K}=\{\mu \in ca\mid \int fd\mu \geq 0\;\forall \,f\in \mathcal{K}%
^{\circ }\}\subseteq ca_{+},
\end{equation*}%
where $\mathcal{K}^{\circ }=-C_{b}^{-}$ are the non decreasing functions $%
f\in C_{b}$.

\begin{definition}[\protect\cite{DK10}]
We denote by $\mathcal{R}^{\max }$ the class of functions $R:\mathbb{R}%
\times \mathcal{K}^{\circ }\rightarrow \overline{\mathbb{R}}$ such that: (i) 
$R$ is non decreasing and left continuous in the first argument,(ii) $R$ is
jointly quasiconcave, (iii) $R(s,\lambda f)=R(\frac{s}{\lambda },f)$ for
every $f\in \mathcal{K}^{\circ }$, $s\in \mathbb{R}$ and $\lambda >0$, (iv) $%
\lim_{s\rightarrow -\infty }R(s,f)=\lim_{s\rightarrow -\infty }R(s,g)$ for
every $f,g\in \mathcal{K}^{\circ }$, (v) $R^{+}(s,f)=\inf_{s^{\prime
}>s}R(s^{\prime },f)$, is upper semicontinuous in the second argument.
\end{definition}

\begin{proposition}
\label{unique} Any $\sigma (\mathcal{P}(\mathbb{R}),C_{b})$-lsc Risk Measure 
$\Phi :\mathcal{P}(\mathbb{R})\rightarrow \mathbb{R}\cup \left\{ \infty
\right\} $ can be represented as in \ref{11}. The function $R^{-}(t,f)$
given by \ref{1111} is unique in the class $\mathcal{R}^{\max}$.
\end{proposition}

\begin{proof}
According to Definition 2.13 in \cite{DK10} a map $\Phi :\mathcal{P}%
\rightarrow \overline{R}$ is continuously extensible to $ca$ if 
\begin{equation*}
\overline{\mathcal{A}^{m}+\mathcal{K}}\cap \mathcal{P}=\mathcal{A}^{m}
\end{equation*}%
where $\mathcal{A}^{m}$ is acceptance set of level $m$ and $\mathcal{K}$ is
the ordering positive cone on $ca$. Observe that $\mu \in ca_{+}$ satisfies $%
\mu (E)\geq 0$ for every $E\in \mathcal{B}_{\mathbb{R}}$ so that $P+\mu
\notin \mathcal{P}$ for $P\in \mathcal{A}^{m}$ and $\mu \in \mathcal{K}$
except if $\mu =0$. \newline
For this reason the lsc map $\Phi $ admits a lower semicontinuous extension
to $ca$ and then Theorem 2.19 in \cite{DK10} applies and we get the
uniqueness in the class $\mathcal{R}_{\mathcal{P}}^{\max }$ (see Definition
2.17 in \cite{DK10}). In addition, $\mathcal{R}^{\max }=\mathcal{R}_{%
\mathcal{P}}^{\max }$  follows exactly by the same argument at the end of
the proof of Proposition 3.5 \cite{DK10}. Finally we notice that Lemma C.2
in \cite{DK10} implies that $R^{-}\in \mathcal{R}^{\max }$ since $\gamma
(m,f)$ is convex, positively homogeneous and lsc in the second argument.
\end{proof}

\subsection{Computation of the dual function}

The following proposition is useful to compute the dual function $R^{-}(t,f)$
for the examples considered in this paper.

\begin{proposition}
Let $\left\{ F_{m}\right\} _{m\in \mathbb{R}}$ be a feasible family and
suppose in addition that, for every $m$, $F_{m}(x)$ is increasing in $x$ and 
$\lim_{x\rightarrow +\infty }F_{m}(x)=1$. The associated map $\Phi :\mathcal{%
P}\rightarrow \mathbb{R}\cup \{+\infty \}$ defined in (\ref{phi1}) is well
defined, (Mon), (Qco) and $\sigma (\mathcal{P},C_{b})-$lsc and the
representation (\ref{11}) holds true with $R^{-}$ given in (\ref{1111}) and 
\begin{equation}
\gamma (m,f)=\int fdF_{-m}+F_{-m}(-\infty )f(-\infty ).  \label{gamma}
\end{equation}
\end{proposition}

\begin{proof}
From equations (\ref{am}) and (\ref{am1}) we obtain:%
\begin{equation*}
\mathcal{A}^{-m}=\left\{ Q\in \mathcal{P}(\mathbb{R})\mid F_{Q}\leq
F_{-m}\right\} =\left\{ Q\in \mathcal{P}\mid \Phi (Q)\leq m\right\}
\end{equation*}%
so that 
\begin{equation*}
\gamma (m,f):=\sup_{Q\in \mathcal{P}}\left\{ \int fdQ\mid \Phi (Q)\leq
m\right\} =\sup_{Q\in \mathcal{P}}\left\{ \int fdQ\mid F_{Q}\leq
F_{-m}\right\} .
\end{equation*}%
Fix $m\in \mathbb{R}$, $f\in C_{b}^{-}$ and define the distribution function 
$F_{Q_{n}}(x)=F_{-m}(x)\mathbf{1}_{[-n,+\infty )}$ for every $n\in \mathbb{N}
$. Obviously $F_{Q_{n}}\leq F_{-m}$, $Q_{n}\downarrow $ and, taking into
account (\ref{cb}), $\int fdQ_{n}$ is increasing. For any $\varepsilon >0$,
let $Q^{\varepsilon }\in \mathcal{P}$ satisfy $F_{Q^{\varepsilon }}\leq
F_{-m}$ and $\int fdQ^{\varepsilon }>\gamma (m,f)-\varepsilon $. Then: $%
F_{Q_{n}^{\varepsilon }}(x):=F_{Q^{\varepsilon }}(x)\mathbf{1}_{[-n,+\infty
)}\uparrow F_{Q^{\varepsilon }}$, $F_{Q_{n}^{\varepsilon }}\leq F_{Q_{n}}$
and%
\begin{equation*}
\int fdQ_{n}\geq \int fdQ_{n}^{\varepsilon }\uparrow \int fdQ^{\varepsilon
}>\gamma (m,f)-\varepsilon .
\end{equation*}%
We deduce that $\int fdQ_{n}\uparrow \gamma (m,f)$ and, since 
\begin{equation*}
\int fdQ_{n}=\int_{-n}^{+\infty }fdF_{-m}+F_{-m}(-n)f(-n),
\end{equation*}%
we obtain (\ref{gamma}).
\end{proof}

\begin{example}
Computation of $\gamma (m,f)$ for the $\Lambda V@R$.

\noindent Let $m\in \mathbb{R}$ and $f\in C_{b}^{-}$. As $F_{m}(x)=\Lambda
(x)\mathbf{1}_{(-\infty ,m)}(x)+\mathbf{1}_{[m,+\infty )}(x),$ we compute
from (\ref{gamma}):%
\begin{equation}
\gamma (m,f)=\int_{-\infty }^{-m}fd\Lambda +(1-\Lambda (-m))f(-m)+\Lambda
(-\infty )f(-\infty ).  \label{gamma1}
\end{equation}%
We apply the integration by parts and deduce 
\begin{equation*}
\int_{-\infty }^{-m}\Lambda df=\Lambda (-m)f(-m)-\Lambda (-\infty )f(-\infty
)-\int_{-\infty }^{-m}fd\Lambda .
\end{equation*}%
We can now substitute in equation (\ref{gamma1}) and get: 
\begin{eqnarray}
\gamma (m,f) &=&f(-m)-\int_{-\infty }^{-m}\Lambda df=f(-\infty
)+\int_{-\infty }^{-m}(1-\Lambda )df,  \label{111} \\
R^{-}(t,f) &=&-H_{f}^{l}(t-f(-\infty )),  \label{222}
\end{eqnarray}
where $H_{f}^{l}$ is the left inverse of the function: $m\rightarrow
\int_{-\infty }^{m}(1-\Lambda )df$. 

As a particular case, we match the results obtained in \cite{DK10} for the $%
V@R$ and the Worst Case risk measure. Indeed, from (\ref{111}) and (\ref{222}%
) we get: $R^{-}\left( t,f\right) =-f^{l}\left( \frac{t-\lambda f(-\infty )}{%
1-\lambda }\right) $ if $\Lambda (x)=\lambda $;  $R^{-}\left( t,f\right)
=-f^{l}(t)$, if $\Lambda (x)=0$, where $f^{l}$ is the left inverse of $f$.

\bigskip

If $\Lambda $ is decreasing we may use Remark \ref{remarkDecr} to derive a
simpler formula for $\gamma $. Indeed, $\Lambda V@R(P)=\Lambda \widetilde{V}%
@R(P)$ where $\forall m\in \mathbb{R}$ 
\begin{equation*}
\widetilde{F}_{m}(x)=\Lambda (m)\mathbf{1}_{(-\infty ,m)}(x)+\mathbf{1}%
_{[m,+\infty )}(x)
\end{equation*}%
and so from (\ref{111}) 
\begin{equation*}
\gamma (m,f)=f(-\infty )+[1-\Lambda (-m)]\int_{-\infty }^{-m}df=[1-\Lambda
(-m)]f(-m)+\Lambda (-m)f(-\infty ).
\end{equation*}
\end{example}


\begin{thebibliography}{99}
\bibitem{Ali} \textsc{Aliprantis C.D., and K.C. Border} (2005): \textit{%
Infinite dimensional analysis}, Springer, Berlin, 3rd edition.

\bibitem{ADEH} \textsc{Artzner, P., Delbaen, F., Eber, J.M., and D. Heath}
(1999): Coherent measures of risk, \textit{Math. Finance}, 4 , 203--228.

\bibitem{BF09} \textsc{Biagini, S., and M. Frittelli} (2009): On the
extension of the Namioka-Klee theorem and on the Fatou property for risk
measures, \textit{In: Optimality and risk: modern trends in mathematical
finance}, Eds: F. Delbaen, M. Rasonyi, Ch. Stricker, Springer, Berlin, 1--29.

\bibitem{CM09} \textsc{Cherny, A., and D. Madan} (2009): New measures for
performance evaluation, \textit{Review of Financial Studies}, 22, 2571-2606.

\bibitem{CMMMb} \textsc{Cerreia-Vioglio, S., Maccheroni, F., Marinacci, M.
and Montrucchio, L.} (2009) \textquotedblleft Complete Monotone Quasiconcave
Duality \textquotedblright , forthcoming on \textit{Math. Op. Res.}

\bibitem{CV} \textsc{Cerreia-Vioglio, S.} (2009): Maxmin Expected Utility on
a Subjective State Space: Convex Preferences under Risk, preprint.

\bibitem{CMMMa} \textsc{Cerreia-Vioglio, S., Maccheroni, F., Marinacci, M.,
and L. Montrucchio} (2011): Risk measures: rationality and diversification, 
\textit{Mathematical Finance}, 21, 743-774.

\bibitem{DK10} \textsc{Drapeau, S., and M. Kupper} (2010): Risk preferences
and their robust representation, preprint.

\bibitem{Fe} \textsc{Fenchel, W.} (1949): On conjugate convex functions, 
\textit{Canadian Journal of Mathematics}, 1, 73-77.

\bibitem{FoSch} \textsc{F$\ddot{o}$llmer, H., and A. Schied} (2004): \textit{%
Stochastic Finance. An introduction in discrete time}, 2nd ed., de Gruyter
Studies in Mathematics, 27.

\bibitem{FM09} \textsc{Frittelli, M., and M. Maggis} (2011): Dual
representation of quasiconvex conditional maps, \textit{SIAM Jour Fin Math},
2, 357--382.

\bibitem{FR05} \textsc{Frittelli, M., and E. Rosazza Gianin} (2005):
Law-invariant convex risk measures, \textit{Advances in Mathematical
Economics}, 7, 33-46.

\bibitem{fr} \textsc{Frittelli, M., and E. Rosazza Gianin} (2002): Putting
order in risk measures, \textit{Journal of Banking \& Finance}, 26(7),
1473-1486.

\bibitem{JST06} \textsc{Jouini, E., Schachermayer, W., and N. Touzi} (
2006): Law invariant risk measures have the Fatou property, \textit{Advances
in Mathem. Economics}, 9, 49-71.

\bibitem{K01} \textsc{Kusuoka, S.} (2001): On law invariant coherent risk
measures, \textit{Advances in Mathematical Economics}, 3, Springer, Tokyo,
83-95.

\bibitem{S10} \textsc{Svindland, G.} (2010): Continuity Properties of
Law-Invariant (Quasi-) Convex Risk Functions, \textit{Math. and Financial
Economics}, 39--43.

\bibitem{Shy} \textsc{Shiryaev, A.N.} (1995): \textit{Probability}, Graduate
Texts in Mathematics, Springer 2nd Edition.

\bibitem{Volle} \textsc{Volle, M.} (1998): Duality for the level sum of
quasiconvex functions and applications, \textit{Control, Optimisation and
Calculus of Variations}, 3, 329-343.

\bibitem{Weber} \textsc{Weber, S.} (2006): Distribution-Invariant Risk
Measures, Information, and Dynamic Consinstency, \textit{Math. Finance},
16(2), 419--441.
\end{thebibliography}
\end{document}